\documentclass[conference]{IEEEtran}

\IEEEoverridecommandlockouts
\usepackage{cite}
\usepackage{amsmath,amssymb,amsfonts}
\usepackage{algorithmic}
\usepackage{graphicx}
\usepackage{textcomp}
\usepackage{xcolor}
\usepackage[a4paper, total={184mm,239mm}]{geometry}
\def\BibTeX{{\rm B\kern-.05em{\sc i\kern-.025em b}\kern-.08em
    T\kern-.1667em\lower.7ex\hbox{E}\kern-.125emX}}

\usepackage{subcaption}

\usepackage{amsthm}
\newtheorem{theorem}{Theorem}[section]

\newtheorem{observation}[theorem]{Observation}

\usepackage{float}
\usepackage{bm}
\usepackage{mathtools}

\usepackage{tabularx,ragged2e}
\usepackage{makecell}
\newcolumntype{C}{>{\Centering\arraybackslash}X} 
\newcolumntype{P}{>{\arraybackslash}X} 
\usepackage{colortbl}

\definecolor{mygray}{gray}{0.7}

\usepackage[ruled,vlined,linesnumbered]{algorithm2e}
\DontPrintSemicolon
\SetKwComment{Comment}{/* }{ */}

\newcommand{\centered}[1]{\begin{tabular}{l} #1 \end{tabular}} 
    
\begin{document}

\title{Optimal Fixed Priority Scheduling in Multi-Stage Multi-Resource Distributed Real-Time Systems\\
\thanks{$\textbf{Funding:}$ This work was supported by the MoE Tier-2 grant MOE-T2EP20221-0006.}
}

\author{\IEEEauthorblockN{Niraj Kumar, Chuanchao Gao, and Arvind Easwaran}
\IEEEauthorblockA{School of Computer Science and Engineering, Nanyang Technological University
Singapore\\
Email: niraj.kumar@ntu.edu.sg, gaoc0008@e.ntu.edu.sg, arvinde@ntu.edu.sg}}

\maketitle

\begin{abstract}

This work studies fixed priority (FP) scheduling of real-time jobs with end-to-end deadlines in a distributed system. Specifically, given a multi-stage pipeline with multiple heterogeneous resources of the same type at each stage, the problem is to assign priorities to a set of real-time jobs with different release times to access a resource at each stage of the pipeline subject to the end-to-end deadline constraints. Note, in such a system, jobs may compete with different sets of jobs at different stages of the pipeline depending on the job-to-resource mapping. To this end, following are the two major contributions of this work. We show that an OPA-compatible schedulability test based on the \emph{delay composition algebra} can be constructed, which we then use with an \emph{optimal priority assignment} algorithm to compute a \emph{priority ordering}. Further, we establish the versatility of \emph{pairwise priority assignment} in such a multi-stage multi-resource system, compared to a total priority ordering. In particular, we show that a pairwise priority assignment may be feasible even if a priority ordering does not exist. We propose an integer linear programming formulation and a scalable heuristic to compute a pairwise priority assignment. We also show through simulation experiments that the proposed approaches can be used for the holistic scheduling of real-time jobs in edge computing systems. 

\end{abstract}



\section{Introduction}\label{sec:intro}

\emph{Fixed priority} (FP) real-time scheduling has attracted significant attention in academia and broader acceptance in the industry due to its predictability and low runtime overhead~\cite{bril2017fixed, davis2016review}.
FP uniprocessor scheduling is NP-hard in general. However, due to the \emph{optimal priority assignment} (OPA) proposed by Audsley~\cite{audsley1991optimal}, the hardness is attributed to the schedulability test. More specifically, OPA is an optimal\footnote{An FP scheduling algorithm $\mathcal{P}$ is said to be \emph{optimal} with respect to a schedulability test $\mathcal{S}$ provided any schedulable taskset by some FP scheduling algorithm (under $\mathcal{S}$) is also schedulable by $\mathcal{P}$.} FP scheduling algorithm to compute a \emph{priority ordering}\footnote{A permutation or a total order of tasks that indicates priority.} of preemptive sporadic tasks with arbitrary release times on a uniprocessor. An optimal priority assignment with OPA only requires $O(n^2)$, for $n$ tasks, invocations to an OPA-compatible\footnote{Discussed, along with delay composition algebra, in Section \ref{sec:background}.} schedulability test. 

Schedulability tests for distributed real-time systems are inherently more complex. Thus, a common approach to analyze the schedulability in distributed systems is to decompose the end-to-end deadline to a per-stage deadline and then independently analyze each stage as a uniprocessor system~\cite{palencia2003offset, hong2011meeting}.
Alternatively, \emph{delay composition algebra} (DCA) is proposed to transform a multi-stage system into an equivalent single-stage system and then analyze schedulability by applying a uniprocessor schedulability test~\cite{jayachandran2008delay, jayachandran2008transforming}. As an intermediate step, a delay composition rule is also presented to compute an upper bound of end-to-end delay for each job in a multi-stage pipeline. In this work, we use the delay composition rule of DCA only; henceforth, the term DCA refers to the same.

This work aims to compute an FP schedule for real-time jobs with release offsets and end-to-end deadlines in a distributed system. In particular, we focus on a multi-stage multi-resource (MSMR) system in which each stage corresponds to a resource type with multiple heterogeneous resources of the same type. Each real-time job accesses one of the resources at every stage in the pipeline and must exit the pipeline within the specified deadline. A simplified version of this problem, particularly multi-stage single-resource scheduling, reduces to flow shop scheduling of jobs with a deadline, which is known to be an NP-hard problem in general~\cite{bettati1992end}. Furthermore, the problem addressed in this work is more complex as we consider a generic model with heterogeneity among resources of the same type and jobs with arbitrary release times. Moreover, jobs are allowed to compete with distinct sets of jobs at different stages of the pipeline depending on the job-to-resource mapping. The major contributions of this work are summarized as follows.

\begin{itemize}
    \item To our knowledge, this work is the first to use OPA and DCA for optimal FP scheduling of real-time jobs with end-to-end deadlines in an MSMR system. Specifically, we show that an OPA-compatible schedulability test, based on DCA, can be constructed to compute an optimal priority ordering\footnote{It is worth mentioning that DCA is not originally intended to compute a priority assignment. Instead, given a priority assignment, DCA bounds the end-to-end delay of each job in a multi-stage system to obtain an equivalent single-stage system for schedulability analysis.} of real-time jobs in an MSMR system. The proposed algorithm can be applied for both preemptive and non-preemptive FP scheduling in an MSMR system for a given job-to-resource mapping.
    \item We establish that \emph{pairwise priority assignment}\footnote{Priorities are assigned to each pair of jobs that share at least one resource in the pipeline as opposed to priority ordering in which a total order is established among all jobs.} and priority ordering are not necessarily equivalent in MSMR systems. In particular, we show that a feasible pairwise priority assignment may exist for a set of real-time jobs that does not admit a priority ordering.
    Further, we propose an ILP formulation and a \emph{deadline-monotonic \& repair} based heuristic to compute an optimal and a sub-optimal pairwise priority assignment, respectively.
    \item Through simulation experiments, we show that the proposed approaches can be applied for holistic scheduling of real-time jobs with end-to-end deadlines in edge computing systems. Notably, these approaches alleviate the challenges associated with computing and imposing a virtual deadline, as end-to-end deadlines are used directly for the priority assignment.
    
\end{itemize}


\section{System Model and Problem Statement} \label{sec:model}

Consider a multi-stage multi-resource (MSMR) system with $N$ stages in which each stage corresponds to a resource of a specific type (for instance, edge servers). 
Moreover, multiple resources of heterogeneous capabilities are available at each stage. Let $\mathcal{J}$ be the given set of $n$ real-time jobs. Each job $J_i$ is specified with the following parameters (i) $A_i$: arrival time, (ii) $P_{i,j}$:  processing time\footnote{We use processing time in a generic sense, thus, for instance, for a networking resource it means transmission time.} of $J_i$ in stage $S_{j}$, (iii) $D_i$: end-to-end deadline, and (iv) $R_{i,j}$: the resource to which $J_i$ is mapped in stage $S_{j}$. 
Let $\mathcal{M}_{i,j}$ be the set of all jobs mapped with $J_i$ in stage $S_j$ (and thus, compete for $R_{i,j}$). Moreover, the set of all jobs that compete with $J_i$ is $\mathcal{M}_i = \cup_{j} \mathcal{M}_{i,j}$. 
%

Let $\Delta_i$ denote the \emph{end-to-end delay} of $J_i$, which is the time it exits the pipeline since its arrival. Given an MSMR system and a set $\mathcal{J}$ of $n$ real-time jobs $J_i  \langle A_i, \{P_{i,j}\}, D_i, \{R_{i,j}\} \rangle$, then the problem of computing a
\begin{itemize}
    \item[P1.] {\fontfamily{lmtt}\selectfont Priority Ordering} is to devise an OPA-compatible schedulability test $\mathcal{S}^{DCA}$ based on DCA and then use it with OPA to assign a unique \emph{global} (i.e., valid across all stages) priority $\rho_i \in [1,n]$ to each job $J_i$ such that the deadline constraint is satisfied $\Delta_i \leq D_i$. 
    \item[P2.] {\fontfamily{lmtt}\selectfont Pairwise Priority Assignment} is to assign a \emph{global} fixed priority $J_i > J_k$ (indicating $J_i$ has a higher priority than $J_k$) or $J_k > J_i$ for each job pair $\langle J_i, J_k \rangle | J_k \in \mathcal{M}_i$ such that $\Delta_i \leq D_i$. 
\end{itemize}
\noindent
If a feasible priority ordering exists, solving P1 is sufficient (since OPA is optimal); however, otherwise, solving P2 is required (as a pairwise priority assignment may still be feasible, further details in Section \ref{sec:problem_p2}). 
%

We define a \emph{segment} corresponding to a job pair $\langle J_i, J_k \rangle$ as one of the longest sequences of consecutive stages for which $J_i$ and $J_k$ are mapped to the same resources. A pipeline may have multiple segments corresponding to each job pair. Let $m_{i,k}$ denote the number of segments corresponding to the job pair $\langle J_i, J_k \rangle$. For instance, $m_{i,b} = 0$ in Figure \ref{fig:proposed_delay}(a), $m_{i,b} = 1$ in Figures \ref{fig:proposed_delay}(b)-(d), and $m_{i,b} = 2$ in Figure \ref{fig:proposed_delay}(e).

Let $\mathcal{H}_i$ (and $\mathcal{L}_i$) denote the set of higher (and lower) priority jobs of $J_i$. Also, define $\mathcal{Q}_i \ = \mathcal{H}_i \cup J_i$. We assume that any job $J_k$ which cannot interfere $J_i$, i.e.,  $[A_k, A_k + D_k]$ and $[A_i, A_i + D_i]$ do not overlap, is already excluded from $\mathcal{H}_i \text{ (and } \mathcal{L}_i)$. 
We assume a lower value (of $\rho_i$) indicates a higher priority. Both preemptive and non-preemptive FP scheduling are considered in this work. Firstly, essential background is discussed in Section \ref{sec:background}. Subsequently, problems P1 and P2 are addressed in Sections \ref{sec:problem_p1} and \ref{sec:problem_p2}, respectively. The commonly used notations are listed in the Table \ref{table:notations}.

\begin{figure}[t]
    \centering
    \includegraphics[width=0.97\linewidth]{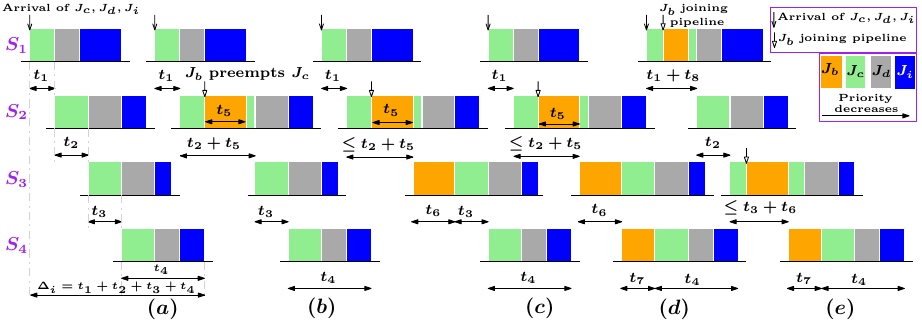}
    \caption{End-to-end delay of $J_i$ in a four-stage pipeline. (a) $\Delta_i = t_1 +  t_2 +  t_3 +  t_4$. (b) One job-additive term due to $J_b$ as it delays $J_i$ at one stage only, $\Delta_i = t_1 +  t_2 + t_5 + t_3 +  t_4$. (c) \& (d) Two job-additive terms if $J_b$ delays $J_i$ for two or more stages. (e) $m_{i,b}=2$ and up to two job-additive terms of $J_b$ corresponding to each segment.}
    \label{fig:proposed_delay}
\end{figure}


\section{Background} \label{sec:background}

\subsection{Delay Composition Algebra (DCA)} \label{sec:dca_intro}

The end-to-end delay bound of a job (henceforth, of $J_i$) in a multi-stage single-resource pipeline\footnote{All jobs compete with each other for every resource in every stage.} can be computed using DCA as explained in the following (reproduced from~\cite{jayachandran2008delay}). 

Let $t_{i,x}$ denote the $x^{th}$ maximum stage-processing time of $J_i$, for instance, $t_{i,1} = \max_{j} \{P_{i,j} \}$. Then, in preemptive scheduling, as illustrated in Figure~\ref{fig:proposed_delay}(a), $\Delta_i$ can be bounded by (i) \emph{stage-additive component:} corresponding to each stage, except the last, to incorporate the overlapped execution of jobs of $\mathcal{Q}_i$, and (ii) \emph{job-additive component}: corresponding to each higher priority job, as it may delay $J_i$ while leaving the pipeline.
Let $\mathcal{H}_i^{a} \subseteq \mathcal{H}_i$ be the set of higher priority jobs that join the pipeline after $J_i$ i.e. a job $J_k \in \mathcal{H}_i^{a}$ provided $J_k \in \mathcal{H}_i \ \text{and} \ A_i~<~A_k$. Then, $J_k \in  \mathcal{H}_i^{a}$ will also contribute a job-additive component while joining the pipeline. Referring to Figure \ref{fig:proposed_delay}(c) for an illustration, $J_b$ contributes two terms, one while joining the execution pipeline of $J_i$ at $S_2$ and another while leaving at $S_3$. 
Thus, combining these interferences, the end-to-end delay of $J_i$ in preemptive scheduling is (Eq.~1,~\cite{jayachandran2008delay})
\begin{equation}\label{eq:dca_orig_pr}
    \Delta_i \leq \sum\limits_{ J_k \in \mathcal{Q}_i} t_{k,1} + \sum\limits_{ J_k \in \mathcal{H}_i^{a}} t_{k,2} +  \sum\limits_{j = 1}^{N - 1} \max_{ J_k \in \mathcal{Q}_i} \{P_{k,{j}}\}
\end{equation}
In non-preemptive scheduling, one lower priority job may delay $J_i$ in each stage; however, jobs $J_k \in \mathcal{H}_i^{a}$ cannot preempt $J_i$. Thus, in this case, the delay of $J_i$ is (Eq.~2,~\cite{jayachandran2008delay})
\begin{equation}\label{eq:dca_orig_npr}
    \Delta_i \leq \sum\limits_{ J_k \in \mathcal{Q}_i}  t_{k,1} +  \sum\limits_{j = 1}^{N - 1} \max_{ J_k \in \mathcal{Q}_i}\{P_{k,{j}}\} + \sum\limits_{j = 1}^{N} \max_{ J_k \in \mathcal{L}_i}\{P_{k,{j}}\}  
\end{equation}

The extended version of DCA~\cite{jayachandran2008transforming} can be used to analyze the delay in an MSMR system with multiple segments for each job pair. We define $\widetilde{p}_{k,{j}} = P_{k,{j}}$ if $J_i$ and $J_k$ share the $j^{th}$ stage, otherwise $\widetilde{p}_{k,{j}} = 0$.
Correspondingly, we also define $\widetilde{t}_{k,x}$, for instance, $\widetilde{t}_{k,1} =  \max_{j} \{\widetilde{p}_{k,{j}} \}$. 
Analogous to Eq. \ref{eq:dca_orig_pr}, a higher priority job $J_k$ contributes two job-additive terms corresponding to each of the $m_{i,k}$ segments in a preemptive scheduling, and thus, the delay is (Eq.~1,~\cite{jayachandran2008transforming})
\begin{equation}\label{eq:dca_orig_pr_multiresource}
    \Delta_i \leq \sum\limits_{ J_k \in \mathcal{Q}_i} 2 m_{i,k} \widetilde{t}_{k,1} + \sum\limits_{j = 1}^{N - 1} \max_{ J_k \in \mathcal{Q}_i}\{\widetilde{p}_{k,{j}}\}
\end{equation}
Likewise, Eq.~\ref{eq:dca_orig_npr} is rewritten to (Eq.~5,~\cite{jayachandran2008transforming}) 
\begin{equation}\label{eq:dca_orig_npr_multiresource}
    \Delta_i \leq \sum\limits_{ J_k \in \mathcal{Q}_i} m_{i,k} \widetilde{t}_{k,1} +  \sum\limits_{j = 1}^{N - 1} \max_{ J_k \in \mathcal{Q}_i}\{\widetilde{p}_{k,{j}}\} + \sum\limits_{j = 1}^{N} \max_{ J_k \in \mathcal{L}_i}\{\widetilde{p}_{k,{j}}\}  
\end{equation}

\subsection{Optimal Priority Assignment (OPA)} \label{sec:opa_intro}

OPA computes a priority ordering of $n$ tasks by assigning priorities from $n$ (lowest) to 1 (highest) in that order. The current priority is assigned to one of the feasible tasks ($\Delta_i \leq D_i$), which is yet to be assigned a priority. If no task is feasible for a priority, the taskset is declared unschedulable. The process is repeated until all tasks are assigned a priority or no feasible task is found for a priority.

A schedulability test is used with OPA to determine the feasibility of a task for a priority. However, the schedulability test must be OPA-compatible~\cite{davis2009priority, davis2016review}. Specifically, a schedulability test is OPA-compatible, provided it satisfies the following three necessary and sufficient conditions. The \textbf{first (second) condition} requires that the schedulability of a task may depend on the set of higher (lower) priority tasks but not on their relative priority. Whereas the \textbf{third condition} states that if the priorities of two tasks with adjacent priorities are swapped, then the task being assigned the lower (higher) priority remains unschedulable (schedulable) if it was unschedulable (schedulable) earlier.
In the following, we use OPA to assign priority to the real-time jobs.

\begin{table}
\small
\centering
\setlength{\extrarowheight}{.15em}
\caption{Notations}
\label{table:notations}
\begin{tabularx}{\linewidth}{cP}
\hline 
Symbol & Meaning  \tabularnewline \hline
$N$ & No. of stages in the pipeline  \tabularnewline  \arrayrulecolor{mygray}\hline
$S_{j}$ & $j^{th}$ stage of the pipeline   \tabularnewline \arrayrulecolor{mygray}\hline
$J_i$ & $i^{th}$ job in the set of jobs $\mathcal{J}$ \tabularnewline  \arrayrulecolor{mygray}\hline
$D_i   / \Delta_i$ & End-to-end deadline/delay of $J_i$ \tabularnewline  \arrayrulecolor{mygray}\hline 
$P_{i,j}$ & Processing time of $J_i$ in stage $S_{j}$  \tabularnewline  \arrayrulecolor{mygray}\hline
$\widetilde{p}_{k,{j}}$ & $P_{k,{j}}$ if $J_i$ and $J_k$ share $S_{j}$, else 0   \tabularnewline \arrayrulecolor{mygray}\hline
$\mathcal{H}_i  / \mathcal{L}_i$ & Set of higher/lower priority jobs of $J_i$  \tabularnewline  \arrayrulecolor{mygray}\hline
$\mathcal{Q}_i$ & Set of higher priority jobs of $J_i$ and itself ($\mathcal{H}_i \cup J_i$)    \tabularnewline  \arrayrulecolor{mygray}\hline
$t_{i,x}$ & $x^{th}$ maximum stage-processing time of $J_i$   \tabularnewline  \arrayrulecolor{mygray}\hline
\centered{$\widetilde{t}_{k,x}$} & $t_{k,x}$ but computed only  for the stages shared with $J_i$ \tabularnewline  \arrayrulecolor{mygray}\hline
$m_{i,k}$ & No. of segments corresponding to job pair $\langle J_i, J_k \rangle$   \tabularnewline  \arrayrulecolor{mygray}\hline
$w_{i,k}$ & No. of job-additive components of $J_k$ that delays $J_i$  \tabularnewline  \arrayrulecolor{mygray}\hline
$\mathcal{M}_{i,j} / \mathcal{M}_i$ & Set of jobs competing $J_i$ in $S_{j}$/in pipeline \tabularnewline  \arrayrulecolor{black}\hline

\end{tabularx}
\end{table}

\section{Optimal Fixed Priority Scheduling in Distributed Real-time Systems} \label{sec:problem_p1}

\subsection{$\mathcal{S}^{DCA}-$ A schedulability test based on DCA} \label{sec:sdca}

We define $\mathcal{S}^{DCA}(J_i, \mathcal{H}_i, \mathcal{L}_i)$ to be a schedulability test based on DCA to determine the feasibility of job $J_i$ given $\mathcal{H}_i$ and $\mathcal{L}_i$.  
However, firstly, we refine the DCA analysis to ensure OPA-compatibility and to reduce the pessimism.

\begin{observation} \label{th:pr_opa_compatible}
    $\mathcal{S}^{DCA}$ based on Eq.~\ref{eq:dca_orig_pr} is OPA-compatible.
\end{observation}

\begin{proof}
    The first two conditions of OPA-compatibility are satisfied trivially as only the set of higher priority jobs is required (and not the relative priority among them) to compute $\Delta_i$.  
    Assume $J_c$ and $J_d$ are two jobs with priorities $\rho_i-1$ and $\rho_i+1$, respectively. Then, the following three cases must be investigated to verify the third condition.
    {\fontfamily{lmtt}\selectfont Case~1.} Priorities of $J_i$ and $J_c$ are swapped. $J_i$ gets a higher priority. As a result, $\mathcal{Q}_i \gets \mathcal{Q}_i \backslash J_c$, and $\mathcal{H}_i^{a} \gets \mathcal{H}_i^{a} \backslash J_c$ if $J_c \in \mathcal{H}_i^{a}$, and thus, $\Delta_i$ would decrease.
    {\fontfamily{lmtt}\selectfont Case~2.} Priorities of $J_i$ and $J_d$ are swapped. $J_i$ gets a lower priority. Thus, $\mathcal{Q}_i \gets \mathcal{Q}_i \cup J_d$ and $\mathcal{H}_i^{a} \gets \mathcal{H}_i^{a} \cup J_d$ if $A_i < A_d$, and hence, $\Delta_i$ would increase.
    {\fontfamily{lmtt}\selectfont Case~3.} Priorities of two jobs with adjacent priorities are swapped such that both are either higher or lower priority jobs than $J_i$. Both $\mathcal{Q}_i$ and $\mathcal{H}_i^{a}$ remains unchanged in this case, and so does $\Delta_i$.
    Thus, the third condition is also satisfied. 
\end{proof}

\begin{observation} \label{th:npr_opa_incompatible}
    $\mathcal{S}^{DCA}$ based on Eq.~\ref{eq:dca_orig_npr} is OPA-incompatible.
\end{observation}

For this observation, consider the following example, for which the third condition of OPA-compatibility is violated.

{\fontfamily{lmtt}\selectfont Example~1.} Consider a distributed system of three stages and four jobs: $J_1 \ \langle 5, 7, 15 \rangle$\footnote{Values inside $\langle \ldots  \rangle$ indicate respective stage-processing times.}, $J_2 \ \langle 7, 9, 17 \rangle$, $J_3 \ \langle 6, 8, 30 \rangle$, and $J_4 \ \langle 2, 4, 3 \rangle$. Assuming $J_1, J_2, J_3, J_4$ to be the priority ordering from highest to lowest, $\Delta_2 = 92$, however if the priorities of $J_2$ and $J_3$ are swapped then $\Delta_2$ reduces to 87. 

Observations \ref{th:pr_opa_compatible} and \ref{th:npr_opa_incompatible} hold for Eqs. \ref{eq:dca_orig_pr_multiresource} and \ref{eq:dca_orig_npr_multiresource}, respectively.
Since $\mathcal{L}_i \subseteq \mathcal{J} \backslash J_i$, we get the following, more pessimistic but OPA-compatible, end-to-end delay for non-preemptive scheduling (by using the set $\mathcal{J} \backslash J_i$ rather than $\mathcal{L}_i$ in the 3$^{rd}$ term of  Eq.~\ref{eq:dca_orig_npr_multiresource}) 
\begin{equation}\label{eq:dca_orig_npr_extended}
    \Delta_i \leq \sum\limits_{ J_k \in \mathcal{Q}_i} m_{i,k} \widetilde{t}_{k,1} +  \sum\limits_{j = 1}^{N - 1} \max_{ J_k \in \mathcal{Q}_i}\{\widetilde{p}_{k,{j}}\} + \sum\limits_{j = 1}^{N} \max_{ J_k \in \mathcal{J} \backslash J_i}\{\widetilde{p}_{k,{j}}\}  
\end{equation}

We also note pessimism in Eq.~\ref{eq:dca_orig_pr_multiresource} as one job-additive term is sufficient for a segment of one stage as the job arrives and leaves at the same stage \big(also illustrated in Figures \ref{fig:proposed_delay}(b) and \ref{fig:proposed_delay}(e)\big). 
Let $J_i$ and $J_k$ share $u_{i,k}$ segments of exactly one stage each and $v_{i,k}$ segments of two or more stages. Then, the maximum number of stage-processing terms of $J_k \in \mathcal{H}_i$ that contributes to $\Delta_i$ is $w_{i,k} = u_{i,k} + 2 v_{i,k}$. Assuming $w_{i,i} = 1$, Eq.~\ref{eq:dca_orig_pr_multiresource} can be rewritten as 
\begin{equation} \label{eq:dca_orig_pr_multiresource_reduced_pessimism}
    \Delta_i \leq \sum\limits_{ J_k \in \mathcal{Q}_i} \sum\limits_{x = 1}^{w_{i,k}} \widetilde{t}_{k,x} + \sum\limits_{j = 1}^{N - 1} \max_{ J_k \in \mathcal{Q}_i}\{\widetilde{p}_{k,{j}}\}
\end{equation}


The delay bound reported by (refined) DCA is used in $\mathcal{S}^{DCA}$ as the value of $\Delta_i$ to determine feasibility. That is to say, $\mathcal{S}^{DCA}$ determines $\Delta_i$ using Eq. \ref{eq:dca_orig_npr_extended} (or Eq. \ref{eq:dca_orig_pr_multiresource_reduced_pessimism}) for non-preemptive (or for preemptive) scheduling, and $J_i$ is deemed feasible under this priority provided $\Delta_i \leq D_i$.

\subsection{Optimal Priority Assignment based on $\mathcal{S}^{DCA}$ (OPDCA)} \label{sec:opdca}

OPDCA, Algorithm \ref{algo:dca_and_opa}, assigns optimal priority (based on OPA) using $\mathcal{S}^{DCA}$. If no task, which is yet to be assigned a priority, is feasible for the current priority level $prior$, in Steps~\ref{step:opdca_inner_for}-\ref{step:opdca_inner_for_ends}, $\mathcal{J}$ is declared infeasible in Step~\ref{stp:opdca_infeasible}, otherwise $prior$ is assigned to a feasible task in each iteration of the outer \emph{for} loop (Steps ~\ref{step:opdca_outer_for}-\ref{stp:opdca_infeasible}).  
In the worst case, for each iteration of the outer \emph{for} loop, the inner \emph{for} loop is executed for all jobs that are yet to be assigned a priority. Thus, $\mathcal{S}^{DCA}$ is called $O(n^2)$ times and the complexity of each call is $O(nN)$. As a result, the overall complexity of Algorithm~\ref{algo:dca_and_opa} is $O(n^3 N)$. 

\begin{algorithm}[t]
\caption{OPDCA}\label{algo:dca_and_opa}
    Set $\mathcal{U} \gets \mathcal{J}$, and $\ell \gets \phi$ \;
    \For{each $prior=n$ \KwTo $1$ \label{step:opdca_outer_for}}{
        $assigned \gets FALSE$ \;
        \For{each $J_i \in \mathcal{U}$ \label{step:opdca_inner_for}}{
            $\mathcal{H}_i \gets \mathcal{U} \backslash J_i$, $\mathcal{L}_i \gets \ell$ \;
            \If{$\mathcal{S}^{DCA}(J_i,\mathcal{H}_i, \mathcal{L}_i)$ is true \label{step:opdca_sdca_call} }{
                $\rho_i \gets prior$; $\mathcal{U} \gets \mathcal{U} \backslash J_i$; $\ell \gets \ell \cup J_i $\;  $assigned \gets TRUE$; {\fontfamily{lmtt}\selectfont Break} \label{step:opdca_inner_for_ends} \; 
            }
        }
        \If{$assigned \text{ is } FALSE$}{
            Declare $\mathcal{J}$ {\fontfamily{lmtt}\selectfont infeasible}  \label{stp:opdca_infeasible}\;
        }
    }
\end{algorithm}

\begin{observation}
    OPDCA is an optimal FP scheduling algorithm to compute a priority ordering with respect to $\mathcal{S}^{DCA}$.
\end{observation}
The proof of this observation follows from the optimality proof of OPA~\cite{audsley2001priority, davis2009priority}. If there exists any FP scheduling algorithm $\mathcal{P}$ that can compute a priority ordering with respect to $\mathcal{S}^{DCA}$, OPDCA can also compute one such priority ordering since in any feasible priority ordering, at least one job is feasible for each priority.
Moreover, OPA is also optimal for non-preemptive scheduling~\cite{davis2016review}; the same holds for OPDCA since $\mathcal{S}^{DCA}$ (based on Eq.~\ref{eq:dca_orig_npr_extended}) is OPA-compatible.


\section{Pairwise Priority Assignment}\label{sec:problem_p2}

A total priority ordering is suitable for single-stage or multi-stage single-resource systems. However, such a total priority ordering in MSMR systems mandates relative priority assignment even to jobs that may not share a single resource in the pipeline. For instance, the priority of $J_1$ with respect to $J_4$ (and vice-versa) is inconsequential for the given job-to-resource mapping shown in Figure~\ref{fig:pairwise_prior2}(a). 

\begin{observation}\label{obs:pairwise}
    A pairwise priority assignment can exist for a set of jobs that does not admit a total priority ordering.
\end{observation}

Consider the distributed system used in {\fontfamily{lmtt}\selectfont Example~1} with deadlines $\{D_i\} = \{60, 55, 55, 50\}$. Furthermore, assume preemptive scheduling and the same arrival time for all jobs. Then, there does not exist a total priority ordering for the job-to-resource mapping shown in Figure \ref{fig:pairwise_prior2}(a); however, a pairwise priority assignment, such as shown in Figure \ref{fig:pairwise_prior2}(b), exists.

\begin{figure}[t]
\centering
\begin{minipage}{.45\linewidth}
    \centering
  \includegraphics[width=0.45\linewidth]{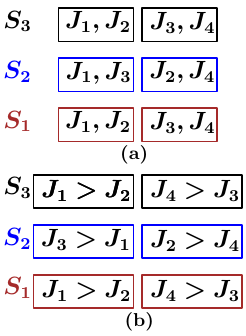}
  \caption{(a). Job-to-resource mapping. (b) Pairwise priority assignment.}
    \label{fig:pairwise_prior2}
\end{minipage}
\hspace{.05\linewidth}
\begin{minipage}{.45\linewidth}
  \includegraphics[width=0.85\linewidth]{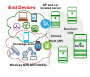}
  \caption{Edge Computing System}
    \label{fig:arch}
\end{minipage}
\end{figure}

The pairwise priority assignment is more flexible; for instance, in Figure \ref{fig:pairwise_prior2}(b), $J_3$ is a lower priority (compared to $J_4$) job at $S_1$, but a higher priority (compared to $J_2$) job at $S_2$. Such flexibility is not viable with priority ordering as, for instance, one job is the lowest priority job at all stages.      
Conclusively, with Observation \ref{obs:pairwise}, it is established that the optimality of OPDCA ceases to exist with respect to problem P2. On the other hand, if a feasible total priority order exists, a pairwise priority assignment is also ensured.

As $\mathcal{S}^{DCA}$ only requires the set of higher (and lower) priority jobs, it can also be used for computing a pairwise priority assignment. 
Furthermore, Eq. \ref{eq:dca_orig_npr_multiresource} can be used to compute a pairwise priority assignment as OPA is not being used (hence, OPA-compatibility is not essential) to assign priority. 

\subsection{ILP Formulation to Compute a Pairwise Priority Assignment}\label{sec:pairwise_ilp} 

In the following, we present an ILP formulation, hereafter referred to as OPT, to compute pairwise priority assignment for preemptive scheduling (the same can also be extended for non-preemptive scheduling).
For every job $J_k \in \mathcal{M}_i$, let $X_{i,k}$ be a binary variable such that $X_{i,k} = 1$ if $J_i > J_k$, otherwise $X_{i,k} = 0$. Thus, 
\begin{equation}\label{eq:complementary_bv}
    X_{i,k} + X_{k,i} = 1
\end{equation}
Moreover, the end-to-end delay of $J_i$ under preemptive scheduling (Eq.~\ref{eq:dca_orig_pr_multiresource_reduced_pessimism}) is 
\begin{equation}\label{eq:delay_eqn_opt}
    \Delta_i = t_{i,1} + \sum\limits_{J_k \in \mathcal{M}_i} X_{k,i} \sum\limits_{x = 1}^{w_{i,k}} \widetilde{t}_{k,x} + \sum\limits_{j = 1}^{N - 1} \theta_{i,j}
\end{equation}
where $\theta_{i,j} = \max_{J_k \in \mathcal{Q}_i}\{\widetilde{p}_{k,{j}}\}$ and computed as described in the following. Let $P = \max_{i,j}\{ P_{i,j} \}$ and $\mathcal{Z}_{i,j}~=~\mathcal{M}_{i,j}~\cup~J_i$. Define $|\mathcal{Z}_{i,j}|$ $\left(|.| \text{ denotes the cardinality of a set}\right)$ auxiliary binary variables $\{ b_y|y \in \left[1, |\mathcal{Z}_{i,j}|\right] \}$. Let $z_{i,j,y} \in \mathcal{Z}_{i,j}$ denote the $y^{th}$ job in $\mathcal{Z}_{i,j}$. 
\begin{subequations}\label{eq:max_computation}
     \begin{align}
          \theta_{i,j} & \geq X_{k,i} \widetilde{p}_{k,{j}}, \forall J_k = z_{i,j,y} | z_{i,j,y} \in \mathcal{Z}_{i,j} \label{eq:max_subeq1}\\
          \theta_{i,j}  & \leq X_{k,i} \widetilde{p}_{k,{j}} + (1 - b_y) P, \forall J_k = z_{i,j,y} | z_{i,j,y} \in \mathcal{Z}_{i,j} \label{eq:max_subeq2}\\
          \sum_y & b_y   = 1 \label{eq:max_subeq3}
     \end{align}
\end{subequations}
Then, the problem of pairwise priority assignment is to determine $X_{i,k}$ subject to constraints (i) specified in Eqs. \ref{eq:complementary_bv} and \ref{eq:max_computation}, and (ii) timing constraints ($\Delta_i \leq D_i | \forall  J_i \in \mathcal{J}$) are satisfied.

\subsection{Deadline-Monotonic \& Repair (DMR) based Heuristic}\label{sec:pairwise_heur}

DMR computes a feasible pairwise priority assignment. DMR begins with a Deadline-Monotonic (DM) based pairwise priority assignment\footnote{DM is not optimal even in a multi-stage single-resource system for preemptive jobs with the same arrival time. The same can be verified by fixing $D_1 = 60$ in Example 1. $J_1$ gets the lowest priority, and thus, $\Delta_1 = 82$.} followed by a \emph{repair} phase in case of timing constraint violation. In the repair phase, for a job pair $\langle J_i, J_k \rangle| J_k \in \mathcal{M}_i, \Delta_i > D_i, \Delta_k < D_k, J_k > J_i $ , the priority is reversed to $J_i > J_k$ provided it is feasible for $J_k$. Pseudo-code for DMR is listed in Algorithm \ref{algo:dmr}. 

\begin{algorithm}[t]
\caption{DMR}\label{algo:dmr}
    \For{each job pair $\langle J_i, J_k \rangle| i \in [1, n-1], k \in [i+1, n],$ $ J_k \in \mathcal{M}_i$}{
        Set $J_i > J_k$ if $D_i \leq D_k$, $J_k > J_i$ otherwise \;
    }
    Compute $\Delta_i, \forall J_i \in \mathcal{J}$ \;
    \For{each job $J_i | \Delta_i > D_i$ \label{step:dmr_outer_for}}{
        Compute $\mathcal{F}_i \subseteq \mathcal{M}_i$ such that $J_k \in \mathcal{F}_i$ provided $J_k > J_i \ \& \ \Delta_k < D_k $ \;
        Sort jobs of $\mathcal{F}_i$ in the decreasing order of $D_k - \Delta_k$\; 
        \For{each job $J_k \in \mathcal{F}_i$ \label{step:dmr_inner_for}}{
            Set $J_i > J_k$ if feasible for $J_k$\;
            Recompute $\Delta_i$. If $\Delta_i \leq D_i$ go to Step \ref{step:dmr_outer_for} \;
        }
        If $\Delta_i > D_i$, declare $\mathcal{J}$ {\fontfamily{lmtt}\selectfont infeasible} \label{stp:dmr_infeasible} \; 
    }
    
\end{algorithm}


\section{Experimental Evaluation}\label{sec:simulation}

The performance of the proposed approaches\footnote{Gurobi has been used to solve OPT. All source codes are available at https://github.com/CPS-research-group/CPS-NTU-Public/tree/DATE2024} is evaluated for holistic scheduling real-time jobs in edge computing systems, a typical example of MSMR systems. 

\subsection{Simulation Setup} 

\textit{System Architecture:}
We consider an edge computing system as shown in Figure \ref{fig:arch}. Deadline-constrained jobs generated by \emph{End Devices} (EDs) are offloaded to an \emph{access point} (AP) through heterogeneous wireless networks using protocols such as 5G-URLLC (Ultra-Reliable Low Latency Communication). APs, in turn, forward the workloads to one of the edge/cloud servers through a wired backhaul network. Subsequently, the computed result is returned to the corresponding ED via an AP. Considering a wired high-speed and contention-free backhaul, we assume that the deadline has already been adjusted to incorporate the backhaul delay.
Thus, the number of stages is $N = 3$ (wireless network for offloading, compute resource at server, and wireless network for downloading). For a given set of jobs and corresponding job-to-resource mapping\footnote{Due to the highly complex nature of holistic resource allocation and scheduling problems, the two problems are commonly addressed separately. Many works, such as \cite{gao2022deadline, chu2022online}, focus on the former. Surveys including \cite{ramanathan2020survey} may be referred for related works.}, the problem is to assign fixed priority to jobs for processing at servers and APs so that timing constraints are satisfied. Furthermore, considering a generic model, preemption is allowed at a server but prohibited during offloading (and downloading) at APs.

In the edge/cloud architecture, the schedule is computed at periodic intervals, and all previously arrived and unscheduled jobs are simultaneously considered. Thus, a lower-priority job begins offloading only after the higher-priority jobs; hence, $\mathcal{H}_i^{a} = \varnothing$. However, a job can be delayed by a lower-priority job during downloading. Thus, Eq.~\ref{eq:dca_orig_pr_multiresource_reduced_pessimism} can be extended by including the delay due to non-preemptive downloading as follows.
\begin{equation} \label{eq:dca_orig_pr_multiresource_reduced_pessimism_edge}
    \Delta_i \leq \sum\limits_{ J_k \in \mathcal{Q}_i} \sum\limits_{x = 1}^{w_{i,k}} \widetilde{t}_{k,x} + \max_{ J_k \in \mathcal{Q}_i}{\widetilde{p}_{k,1}} + \max_{ J_k \in \mathcal{Q}_i}{\widetilde{p}_{k,2}} + \max_{ J_k \in \mathcal{L}_i}{\widetilde{p}_{k,3}}  
\end{equation}
Being a 3-stage system with non-preemption only at the last stage, Eq. \ref{eq:dca_orig_pr_multiresource_reduced_pessimism_edge} is OPA-compatible even if the 4$^{th}$ term (delay component due to non-preemption) is computed over $\mathcal{L}_i$. 

\textit{System Parameters:} The system parameters are selected in line with existing works~\cite{gao2022deadline, chu2022online}. The number of APs (for offloading and downloading), servers, and jobs are fixed to $25$, $20$, and $100$, respectively. Moreover, the uplink and downlink speeds of APs, computation speed of servers, and job parameters are set such that offloading, processing, and downloading times (in milliseconds) of a job are in the ranges $[2,200]$, $[50,500]$, and $[2,100]$, respectively.

We define heaviness of a job $J_i$ at stage $S_{j}$ as $h_{i,j} = \frac{P_{i,j}}{D_i}$. Let $\chi_{y,j}$ be the sum of the heaviness of all jobs mapped to the $y^{th}$ resource at $S_{j}$. Then, we define the heaviness of the set of jobs as $H = \max_{y,j}\{\chi_{y,j}\}$. Note that $H$, with the number of stages, resembles total utilization. Then, the performance of the proposed approaches is reported for varying: (i) \emph{heaviness threshold} ($\beta$): $J_i$ is referred to as \emph{heavy} at $S_{j}$ provided $h_{i,j} \geq \beta$, (ii) \emph{per-stage heaviness requirement} $\left(  [ h_1, h_2, h_3 ] \right)$: $h_{j}$ denotes the ratio of jobs that are heavy at $S_{j}$, and (iii) \emph{heaviness bound} ($\gamma$) such that  $H \leq \gamma$. The maximum heaviness of a job at any stage is always bounded by $2 \beta$. Furthermore, the default values are $\beta = 0.15$, $h_1 = h_2 = 0.05$, $h_3 = 0.01$, and $\gamma = 0.7$. 




To investigate the usefulness of the repair phase in DMR, we use a baseline, referred to as DM, which is like DMR but without any repair phase. A decomposition-based approach, referred to as DCMP, is used as another baseline, which involves decomposing the end-to-end deadline of each job to a virtual deadline corresponding to each stage. 
Let $\Upsilon_{i,j}$ denotes the sum of heaviness of all jobs mapped to $R_{i,j}$ (the resource to which $J_i$ is mapped in stage $S_{j}$); then, the virtual deadline of $J_i$ at $S_{j}$ is $D_i \times \frac{\Upsilon_{i,j}}{\sum_{j} \Upsilon_{i,j}}$.
%
As a schedulability test that can be applied even to the decomposed jobs in this setting does not exist, we simulate the processing of decomposed jobs by assigning priorities in the inverse order of the deadline.

\begin{figure*}
    \centering
    \begin{subfigure}{0.23\textwidth}
        \centering
        \includegraphics[width=0.9\linewidth]{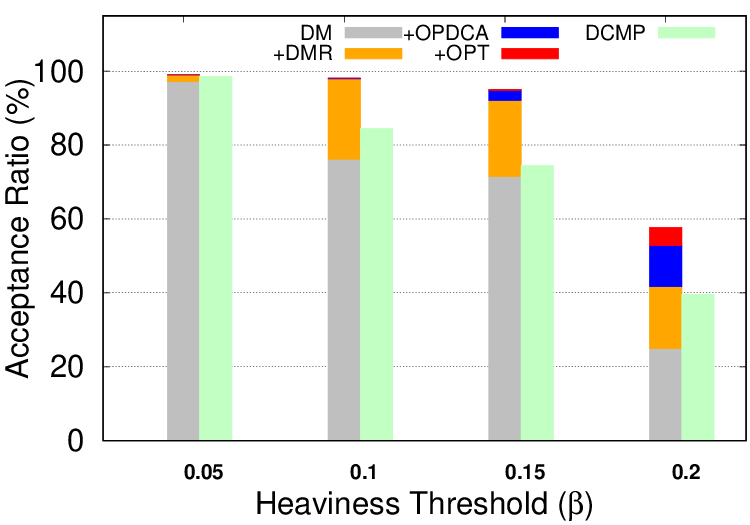}
        \caption{Varying $\beta$}
        \label{fig:ratio_beta}
    \end{subfigure}%
    ~ 
    \begin{subfigure}{0.23\textwidth}
        \centering
        \includegraphics[width=0.9\linewidth]{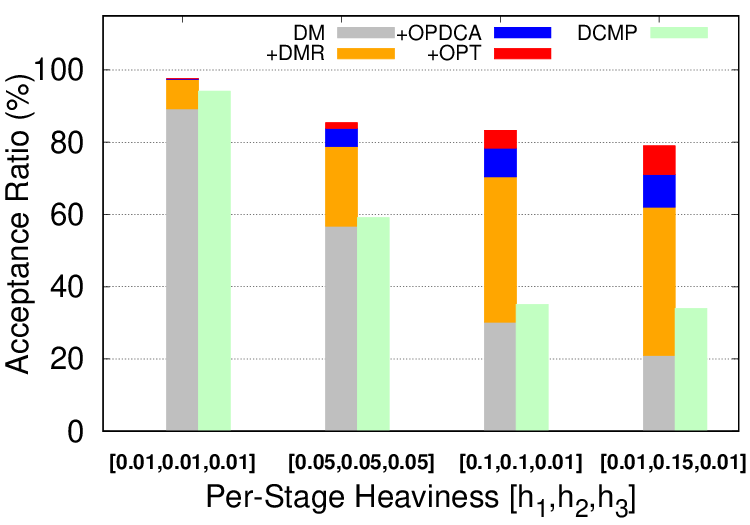}
        \caption{Varying $[h_1,h_2,h_3]$}
        \label{fig:ratio_perstage}
    \end{subfigure}%
    ~ 
    \begin{subfigure}{0.23\textwidth}
        \centering
        \includegraphics[width=0.9\linewidth]{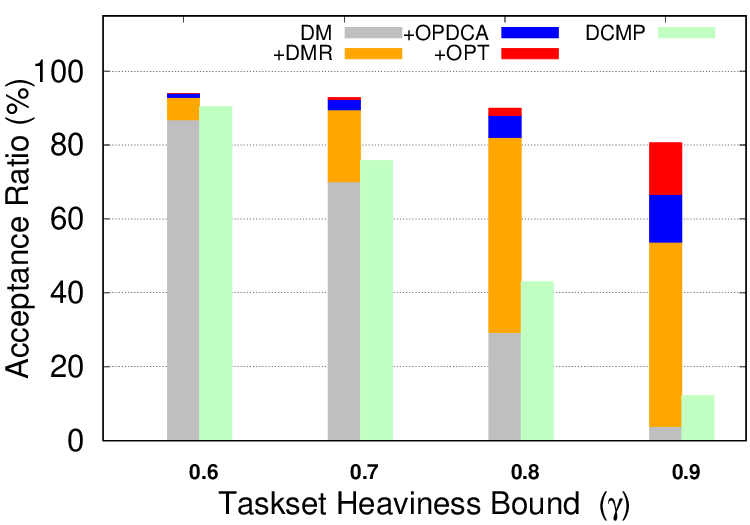}
        \caption{Varying $\gamma$}
        \label{fig:ratio_gamma}
    \end{subfigure}
    ~
    \begin{subfigure}{0.23\textwidth}
        \centering
        \includegraphics[width=0.9\linewidth]{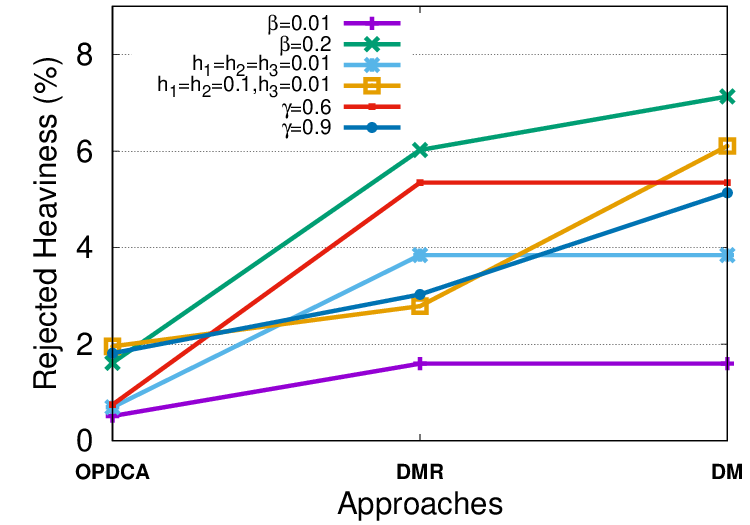}
        \caption{Rejected Heaviness}
        \label{fig:partial}
    \end{subfigure}
    \caption{Acceptance ratios (ARs) for varying (a) $\beta$, (b) $[h_1, h_2, h_3]$, and (c) $\gamma$. The base of the stacked histogram represents the AR of DM. The increment in the AR of DMR, OPDCA, and OPT compared to the AR of DM, DMR, and OPDCA, respectively, are stacked (upward) in this sequence. (d) Performance as admission controller.}
    \label{fig:sim}
\end{figure*}

\subsection{Results and Discussion} 

We define \emph{acceptance ratio} (AR) of an approach as the percentage of accepted test cases. Moreover, AR$_{DM}$, for instance, denotes the AR of DM.
As shown in Figure \ref{fig:sim}, for a lightly loaded system such as $\beta = 0.05$ (i.e., $J_i$ with $P_{i,j} \geq 0.05\times D_i$ is heavy at $S_{j}$), AR$_{DM}$ is substantially high ($>97\%$). However, as $\beta$ increases, the processing time of all jobs, heavy as well as those which are not heavy, increases; the same can be observed with the rapidly deteriorating AR$_{DM}$. However, AR$_{OPDCA}$ (feasibility of computing a priority ordering) remains relatively stable except for heavily loaded systems. For instance, for $\beta = 0.15$, a priority ordering can be computed (by OPDCA) for $94\%$ of the test cases, whereas $71\%$ and $74\%$ test cases are schedulable by DM and DCMP, respectively. AR$_{OPDCA}$ is only marginally higher than AR$_{DM}$ or AR$_{DCMP}$ for lightly loaded system such as $\beta = 0.05$ (or  $\gamma = 0.6$ or $h_j = 0.01$). However, a substantial gap is observed with increasing load.

Likewise, for a lightly loaded system, AR for pairwise priority assignment using DMR is comparable to that of the DM. However, AR$_{DM}$ decreases rapidly with increasing load, while a relatively stable AR$_{DMR}$ indicates the usefulness of the repair phase. We observe a higher AR for computing a priority ordering (by OPDCA, an optimal algorithm) compared to computing a pairwise priority assignment (by DMR, a heuristic). However, AR$_{OPT}$ (to compute an optimal pairwise priority assignment) is higher than AR$_{OPDCA}$, supporting Observation \ref{obs:pairwise}. For a moderate heaviness, such as $\gamma = 0.8$, a pairwise priority assignment can be computed for $90\%$ of the test cases by OPT compared to the priority ordering for $88\%$ by OPDCA. However, the gap in the performance of OPDCA and OPT becomes more evident with higher load. Moreover, the simulation results also reveal the rapidly degrading performance of DCMP with increasing load.


We also observe the performance of the proposed approaches by running them as admission controllers to accept as many jobs as possible. Thus, in Step \ref{stp:opdca_infeasible} of Algorithm \ref{algo:dca_and_opa}, rather than declaring the entire set of jobs infeasible, a job (with the largest $\Delta_i - D_i$) is discarded, and priority assignment is reattempted for the jobs which are yet to be assigned a priority. Likewise, in Step \ref{stp:dmr_infeasible} of Algorithm \ref{algo:dmr}, the job with the largest $\Delta_i - D_i$ is discarded. Figure \ref{fig:partial} illustrates the \emph{rejected heaviness}, defined as the percentage of the heaviness of rejected jobs with respect to the heaviness of all jobs. In general, even for the parameters corresponding to high system load, OPDCA performs reasonably well. 

Conclusively, we observe higher acceptance ratios for OPT and OPDCA compared to other approaches. However, the performance gap increases with the system load, indicating the suitability of pairwise priority assignment. Further, this gap is likely to grow with the number of stages, resources, and jobs. 


\section{Conclusion and Future Works} \label{sec:conclusion}

In this work, the delay composition rule proposed in~\cite{jayachandran2008delay, jayachandran2008transforming} is refined and shown to be OPA-compatible for optimal FP scheduling of real-time jobs with end-to-end delay in a multi-stage multi-resource (MSMR) distributed system. 
Furthermore, the suitability of pairwise priority assignment is also established in MSMR systems. Moreover, an ILP formulation and a deadline-monotonic repair-based heuristic have also been presented to compute a pairwise priority assignment. 
Finally, simulation results are discussed for the holistic scheduling of real-time jobs for compute and network resources in an edge-computing system using the proposed approaches. 

This work presents a starting point for future research to apply and investigate the use of DCA for scheduling, in addition to the originally intended schedulability analysis, in distributed systems. Future research directions also include further exploration of pairwise priority assignment strategies, including developing algorithms with performance guarantees.

\bibliographystyle{IEEEtran}
\bibliography{date}

\end{document}